\documentclass{amsart}

\usepackage{amsmath, amssymb}
  \usepackage{paralist}
  \usepackage{graphicx}
  \usepackage{epsfig}
  \usepackage{epstopdf}

 \usepackage[colorlinks=true]{hyperref}

\newtheorem{theorem}{Theorem}[section]

\newtheorem{lemma}[theorem]{Lemma}
\newtheorem{proposition}{Proposition}

\theoremstyle{definition}

\newcommand{\ep}{\varepsilon}

\newcommand{\abs}[1]{\lvert #1 \rvert}
\newcommand{\bigabs}[1]{\left\lvert #1 \right\rvert}
\newcommand{\lrpar}[1]{\left( #1 \right)}

\numberwithin{equation}{section}
\numberwithin{theorem}{section}

\title[Landau equation and weak coupling limit]
	{Rigorous derivation of the Landau equation in the weak coupling limit}

\author[Kay Kirkpatrick]{}


\email{kay@math.mit.edu} 
\thanks{This work was supported by an NSF Graduate Research Fellowship and an AAUW American Dissertation Fellowship.}

\begin{document}

\maketitle


\centerline{\scshape Kay Kirkpatrick	 }
\medskip
{\footnotesize
 \centerline{Massachusetts Institute of Technology}
   \centerline{77 Mass. Ave., Cambridge, MA 02139, USA}
}

\medskip

\begin{abstract}

We examine a family of microscopic models of plasmas, with a parameter $\alpha$ comparing the typical distance between collisions to the strength of the grazing collisions. These microscopic models converge in distribution, in the weak coupling limit, to a velocity diffusion described by the linear Landau equation (also known as the Fokker-Planck equation). The present work extends and unifies previous results that handled the extremes of the parameter $\alpha$, for the whole range $(0, 1/2]$, by showing that clusters of overlapping obstacles are negligible in the limit. Additionally, we study the diffusion coefficient of the Landau equation and show it to be independent of the parameter.

\end{abstract}

\section{Introduction}

Particles in a plasma experience grazing collisions because they are ionized, interacting even at long distances as described by the Coulomb potential. So far, the full Coulomb model has been impossible to handle rigorously in a scaling limit (see \cite{L} for a heuristic argument, and \cite{DP} for a rigorous partial result), and the strategy has been to use an approximation by soft-sphere models with their bump-function potentials.

Microscopically, soft-sphere models consist of a lightweight particle traveling through a random configuration of large stationary particles, called ``obstacles" or ``scatterers," whose shape and density are determined by a parameter. The lightweight particle grazes obstacles when it gets within their ranges of influence, called ``protection" disks. In two dimensions, this can be visualized as a ball rolling through a random field of hills. It is desirable to understand these microscopic models in the weak coupling limit, when the radius of the obstacles goes to zero, and to derive rigorously a macroscopic description in terms of a linear PDE.
  
\begin{figure}[htpb]
      \centering
      \includegraphics[viewport=0in 0in 5in 3in, keepaspectratio, width=3.5in,clip]{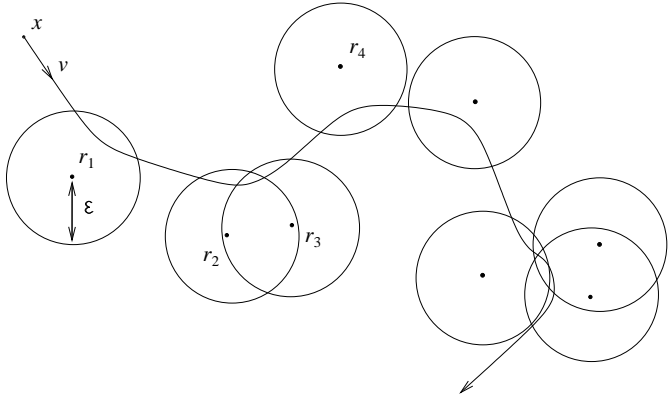}
      \caption{Example of the soft-sphere model: a configuration of obstacles and the corresponding trajectory of the light particle.}
      \label{softspheres}
\end{figure}

More precisely, we introduce a parameter $\alpha \in (0, 1/2]$.  The $i$-th obstacle, centered at $r_i$ (see figure \ref{softspheres}), is described by a suitably smooth and compactly supported (say, in the unit ball) radial potential $V$ whose rescaling is:   
\begin{equation*}
V_{\alpha}^\ep(x-r_i) = \ep^{\alpha} V( |x-r_i|/\ep ).
\end{equation*}
We use the convention of \cite{DR}; another convention for rescaling has $\ep^2$, as in \cite{DGL}.

Then the microscopic dynamics for an obstacle configuration, $\omega$, are Newtonian: 
\begin{equation}\label{Newton}
 \left\{ 
\begin{array}{c}
	\dot{x_\alpha^\ep} = v_\alpha^\ep,   \quad\quad\quad\quad\quad\quad\quad\quad\quad\quad\quad\quad\quad\quad x_\alpha^\ep(0) = x,\\
	 \dot{v_\alpha^\ep} =  - {\displaystyle \sum_{r_i \in \omega} \nabla V_\alpha^\ep (x_\alpha^\ep - r_i) }, 
	\quad\quad\quad\quad\quad\quad\;\, v_\alpha^\ep(0)= v.
\end{array}
\right.
\end{equation}

The collection of obstacle centers, $\omega := \{ r_i : i \in \mathbb{Z} \}$, is a realization of the Poisson point process in $\mathbb{R}^2$ with intensity $\rho_\alpha^\ep := \ep^{-2\alpha- 1} \rho$, e.g., the expected number of obstacles in $A \subset \mathbb{R}^2$ is $\mathbb{E}^\ep (N (A)) = \rho \abs{A} \ep^{-2\alpha- 1}.$

The case $\alpha = 0$, where the particle travels a relatively long distance between collisions and each obstacle has a relatively large influence, corresponds to the Boltzmann-Grad limit of the hard-sphere model of a dilute (or Lorenz) gas, with the macroscopic evolution given by the Boltzmann equation. (See, for instance, \cite{Spohn} and \cite{BBS}; and for the quantum Lorenz gas, \cite{EY1}.)

For $\alpha \in (0,1/2]$, the number of obstacles must be about $\ep^{-2\alpha- d + 1}$ per unit of volume in order to have a net effect that is nonzero and finite (this is more than in the Boltzmann-Grad limit, where the number is about $\ep^{1-d}$). This large number of obstacles balances the small factor of $\ep^{\alpha}$ in the rescaled potential, and as $\ep \to 0$ in this, the weak coupling limit, the macroscopic dynamics of these plasma models are given by the linear Landau equation.


Previously, the results for this family of models were incomplete. Desvillettes and Ricci proved a weak version (convergence in expectation) of the two-dimensional weak coupling limit for $0 < \alpha< 1/8$ \cite{DR}. They approximated the Landau equation by the Boltzmann equation in order to use a modification of Gallavotti's technique for the Boltzmann-Grad limit.  From the outset, however, they needed $\alpha$ to be small, so that the radius of each obstacle is much smaller than the expected free flight time. The method is further limited to the regime $0 < \alpha< 1/8$ by an estimate of the probability of self-intersection, but can be improved to include $0 < \alpha< 1/4$ (see Appendix).  

On the other hand, Kesten and Papanicolaou proved a stronger convergence (in law) of the weak coupling limit at the upper endpoint of the parameter range, $\alpha = 1/2$, for dimensions three and higher and with a general random field, called the stochastic acceleration problem \cite{KestPapSA}. Then D\"urr, Goldstein, and Lebowitz proved a two-dimensional version, with the Poisson distribution of obstacles \cite{DGL}. More recently, Komorowski and Ryzhik handled the stochastic acceleration problem in two dimensions \cite{KomRyzSA}. (See also \cite{EY2} and \cite{HLW} for the quantum weak coupling limit.)

The difficulty in the middle range of $\alpha$ is that obstacles overlap a great deal more than for small $\alpha$, but their individual influence is greater than for $\alpha = 1/2$. There's an additional difficulty in the two-dimensional case because the probability of self-intersections is nontrivial. Fortunately, it turns out that ``bad" self-intersections (ones that are repeated or almost tangential) are negligible in the scaling limit. This is important because such self-intersections cause a correlation between the past and the present (see Figures \ref{almosttang} and \ref{manyselfint}), and controlling memory effects is the main difficulty in these problems.  In higher dimensions, this is unnecessary, because the probability of any self-intersection is negligible in the limit.

\vspace{-0.25in}
\begin{figure}[htpb]
      \centering
      \includegraphics[viewport=0in 0in 4.5in 3.6in, keepaspectratio, width=3.0in,clip]{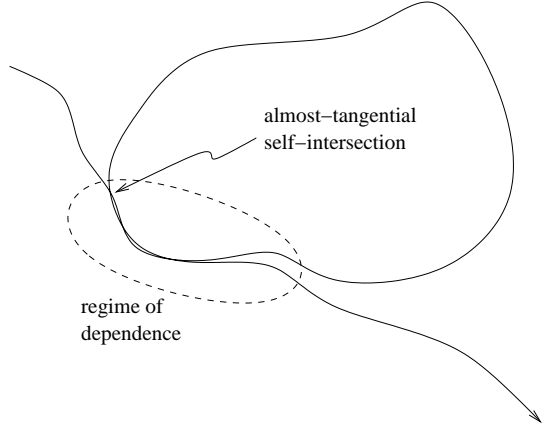}
      \caption{Correlation within a trajectory due to a small-angle self-intersection.}
      \label{almosttang}
\end{figure}

\vspace{-0.35in}
\begin{figure}[htpb]
      \centering
      \includegraphics[viewport=0in 0in 4in 2.5in, keepaspectratio, width=2.9in,clip]{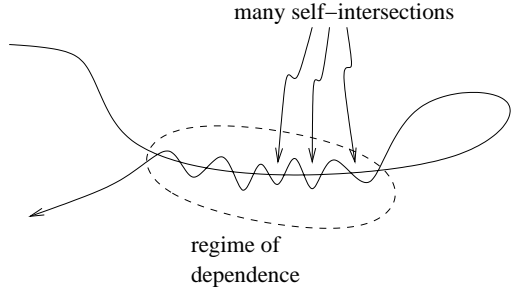}
      \caption{Correlation within a trajectory due to many self-intersections.}
      \label{manyselfint}
\end{figure}

\begin{theorem}
For $\alpha \in (0, 1/2)$, the family of stochastic processes $(v_\alpha^\ep(t))_{t \geq 0}$ converges as $\ep \to 0$ in law to the velocity diffusion $(v(t))_{t\geq0}$ generated by $\Delta_v$, the Laplace-Beltrami operator on $S^1 := \{ w \in \mathbb{R}^2 : \abs{w} = \abs{v_0} \}$.   

In particular, if $f_0(x,v)$ is an initial distribution of positions and velocities, and $\Phi^{t}_{\alpha, \omega, \ep}$ is the flow for the microscopic dynamics \eqref{Newton}, then
\begin{equation}
f^\ep_\alpha(t,x,v) = \mathbb{E} [ f_0(\Phi^{t}_{\alpha, \omega, \ep} (x,v)) ] \xrightarrow{\ep \to 0} h(t,x,v),
\end{equation}
where $h$ is the solution of the linear Landau equation:
\begin{equation}\label{LinearLandau}
\left\{
\begin{aligned}
	(\partial_t + v \cdot \nabla_x) h(t,x,v) & = \zeta \Delta_v h(t,x,v); \\
	h(0,x,v) & = f_0(x,v).
\end{aligned}
\right.
\end{equation}
\end{theorem}

Moreover, the microscopic distinctions of the obstacles' steepness and density all disappear in the scaling limit, and the models have the same macroscopic behavior:

\begin{proposition}
The diffusion coefficient in \eqref{LinearLandau}, $\zeta$, is independent of $\alpha \in (0, 1/2]$ and can be expressed by the following formula:
\begin{equation}\label{zeta}
\zeta = \frac{\rho}{2} \int_{-1}^1 \left( \int_b^1 V' \left( \frac{\abs{b}}{u} \right) \frac{b}{u} \frac{du}{\sqrt{1-u^2}} \right)^2 db.
\end{equation}
\end{proposition}

The main novelties in the present paper are better estimates on clusters of overlapping obstacles and the amount of time spent interacting with each cluster. By controlling these quantities, we show that the total influence of clusters is negligible in the scaling limit. The outline of the paper is as follows:

\begin{itemize}
\item To prove convergence in law, we must show that the measures $\nu^\ep_\alpha$ induced by the microscopic processes $(v_\alpha^\ep(t))$ on $C( [0,T]^2; \mathbb{R}^2 \times \mathbb{R}^2)$ converge weakly to $\nu = \nu_\alpha$ induced by the diffusion $(v(t))$.
\begin{itemize}
\item To this end, we first define some stopping times to eliminate wild behavior (e.g., too many self-crossings, or at too small an angle), which is negligible in the limit (Section 2). 
\item Next we prove that the family of stopped processes is tight (Section 3). 
\item And we identify the limit as a velocity diffusion (Section 4). 
\end{itemize}
\item Then we examine, in Section 5, the particular case of convergence in expectation, using a Gallavotti-type method and a combinatorial argument about clusters of obstacles. 
\item Finally, in Section 6, we prove the Proposition about the diffusion coefficient being constant in $\alpha$, and show that the formulas of \cite{DR} and \cite{DGL} agree. 
\item The Appendix contains a narrower extension of the techniques of \cite{DR}.
\end{itemize}
  
\section{The cut-offs}

We prove the theorem for the processes with stopping times that can be removed for the limiting process. These cut-offs prevent wild behavior that would lead to a correlation (the regimes of dependence in figures \ref{almosttang} and \ref{manyselfint}) between the past of a trajectory and its future. First, for $p \in C([0,T]^2; \mathbb{R}^2 \times \mathbb{R}^2)$, define:
\begin{equation}
Q(t) := \int_0^t p(u) du.
\end{equation}
Then we define the stopping times as follows. The first time the trajectory approaches its past within distance $a$ and with angle less than $\phi$ is cut off by $\tau_{\phi, a}$:
\begin{equation}
\begin{split}
 \tau_{\phi, a} :=  \inf \{ t:  & \exists s \in [0,t] \text{ such that } \abs{Q(s) - Q(t)} \leq a, \\
 & \min_{u \in [s,t]} p(s)\cdot p(u) \leq 0, \text{ and } \frac{ \abs{p(t) \cdot p(s)}}{ \abs{p(t)} \abs{p(s)}} \geq \cos \phi \}.
 \end{split}
\end{equation}
The first time trajectory crosses itself more than $K$ times is cut off by $\tau_K$:
\begin{equation}
\begin{split}
 \tau_{K} := \inf \{ t \geq 0 : & \exists s_i < t_i, \text{ for } i = 1, \dots, K, \text{ and } t_1< \cdots < t_K = t, \\
 & \text{ such that } Q(s_i) = Q(t_i) \forall i \}
\end{split}
 \end{equation}
Additionally there is a velocity cut-off, to prevent the trajectory getting "stuck" somewhere or going too fast:
\begin{equation}
\tau_v := \inf \{ t \geq 0 : \abs{\abs{p(t)} - v_0} \geq v_0/2 \}.
\end{equation}

We call the overall stopping time $\tau$:
\begin{equation}
\tau := \min \{ \tau_{\phi, a}, \tau_K, \tau_v \}.
\end{equation}
These stopping times actually depend on $\ep$, so we write the stopped processes as:
 \begin{equation}
 v^{\ep, \alpha}_t := v^\ep_\alpha( t \wedge \tau^\ep).
 \end{equation} And their induced measures are written as $\tilde{\nu}^\ep_\alpha$.
To see that the cut-offs can be removed, namely that  $\tilde{\nu}^\ep_\alpha$ and the measure induced by the original process, $\nu^\ep_\alpha$, have the same limit, $\nu$, we claim that, for arbitrary $T > 0 $: 
\begin{gather}\label{stoppingtimeproperties}
 \lim_{\ep \to 0} \nu^\ep_\alpha (\tau < T) = 0; \\
  \lim_{\phi, a \to 0} \nu ( \tau_{\phi,a} < T ) = 0; \\
 \lim_{K \to \infty} \nu ( \tau_K \leq T ) = 0.
 \end{gather}
 
The first statement is immediate. The second follows from a self-crossing lemma in \cite[pp. 228-9]{DGL}, which says roughly that tangential self-intersections are negligible. The third follows as a corollary, because the trajectory is almost surely a continuously differentiable curve. 

\section{Tightness}

We will show that the family of processes $v^{\ep, \alpha}_t$ is tight, to get the weak convergence of the induced measures $\tilde{\nu}^\ep_\alpha$. First we need some generalizations of the lemmas in \cite[$\S$ 4]{DGL}.

\subsection{The Martingale-compensator decomposition}

We can view the Poisson point process of obstacles from the particle's perspective (a locally Poisson process denoted $N^\ep_\alpha$, realizations of which are measures on $S \times \mathbb{R}_+$), and the difference (denoted $M^\ep_\alpha$) turns out to be a martingale. This is expressed in the following equation, where the first term on the right-hand side is a martingale, and the second is a measurable left-continuous process called the compensator.
\begin{equation}\label{MGdecomp}
 \int_S \int_0^t N^\ep_\alpha (d\sigma, du) f(\sigma, u) = \int_S \int_0^t M^\ep_\alpha (d\sigma, du) f(\sigma, u) + \int_S \int_0^t \rho^\ep_\alpha (d\sigma, du) f(\sigma, u).
 \end{equation}

Then we write the potential as $F^\ep_\alpha(y):= - \sum \nabla V^\ep_\alpha(y)$, the stopped position process $x^{\ep, \alpha}_t = x^\ep_\alpha( t \wedge \tau^\ep)$, and $x^{\ep, \alpha}_{u, \sigma} (t) := x^\ep_t - x^\ep_u + \sigma \ep$. 
 It is convenient to express $v^{\ep, \alpha}_t$  in terms of $N^\ep_\alpha$,: 
\begin{equation}\label{ptprocessdiff}
 v^{\ep, \alpha}_t - v_0 = \int_{\mathbb{R}^2} \int_0^{t \wedge \tau^\ep} N^\ep_\alpha (d\sigma, du) \int_u^{t \wedge \tau^\ep} F^\ep_\alpha (x^{\ep, \alpha}_{u, \sigma} (s)) ds + \Delta v^{\ep, \alpha}_0 (t).
 \end{equation}

\subsection{Bounding clusters of overlapping obstacles}

To control the probability of $\lambda$ scatterers being found within distance $n\ep$ of the particle at time $t$, we introduce some notation:  $$S^\ep_n(x) := B(x, n \ep),$$
 $$\mathcal{N}^{\ep, \alpha}_n (x) := \abs{ \omega \cap S^\ep_n (x) },$$
 $$B := \sup \abs{ V(x) }.$$ 

\begin{lemma}\label{OL}
For any $\lambda, \ep, T > 0$, 
$$ P^\ep ( \{ \sup_{ t \leq T } \mathcal{N}^{\ep, \alpha}_n (x^{\ep, \alpha}_t) \geq \lambda \} )\leq \lrpar{ \frac{1+ 4\lambda \ep^\alpha BT}{n \ep} + 1 }^2 e^{ 32 n^2 \rho \ep^{-2\alpha- 1} - \lambda }.$$
\end{lemma}

\begin{proof}
The proof follows \cite{DGL} with appropriate modifications. We start from a result of the conservation of energy, with $v_0 = 1$:
$$ (v^{\ep, \alpha}_t)^2 \leq 1 + 4 \ep^\alpha\sup_{s \leq t} \mathcal{N}^{\ep, \alpha}_1 (x^{\ep, \alpha}_s) B. $$

Hence $v^{\ep, \alpha}_t \leq 1 + 4 \ep^\alpha\lambda B$, if $\lambda$ is a bound on $\mathcal{N}^{\ep, \alpha}_1 (x^{\ep, \alpha}_s).$

Then, if we set $\tau_\lambda := \inf \{ t \geq 0: \mathcal{N}^{\ep, \alpha}_n(x^{\ep, \alpha}_t) \geq \lambda \},$ we have
$$ \{ \tau_\lambda > T \} \subset \{ \sup_{t \leq T} v^{\ep, \alpha}_t \leq 1 + 4 \lambda \ep^\alpha B \}
\subset \{ \sup_{t \leq T} x^{\ep, \alpha}_t \leq (1 + 4 \lambda \ep^\alpha B) T \}.$$
That is to say, if the particle does not come close to $\lambda$ scatterers at once before time $T$, then $v^{\ep, \alpha}_t  \leq 1 + 4 \lambda \ep^\alpha B$ for all $t$ before $T$; hence $x^{\ep, \alpha}_t$ is bounded by $(1 + 4 \lambda \ep^\alpha B) T$.

Now we consider the square $[ - (1 + 4 \lambda \ep^\alpha B) T, (1 + 4 \lambda \ep^\alpha B) T ]^2$, which we call $\Gamma$.  If the particle is close to $\lambda$ scatterers simultaneously before time $T$, then
$$ \sup_{x \in \Gamma} \mathcal{N}^{\ep, \alpha}_n (x) \geq \lambda.$$

Tile $\Gamma$ by squares $\gamma^\ep_i$ with side length $2 n \ep^\alpha$ and $i = 1, \dots , [(1 + 4 \lambda \ep^\alpha B) T / n \ep + 1]^2$.  Then
$$ \{ \sup_{x \in \Gamma} \mathcal{N}^{\ep, \alpha}_n (x) \geq \lambda \} \subset \{ \sup_i \sup_{x \in \gamma^\ep_i} \mathcal{N}^{\ep, \alpha}_n (x) \geq \lambda \}.$$

For each $i$, cover $\gamma^\ep_i$ symmetrically by the larger square $\tilde{\gamma}^\ep_i$ of side-length $4 n \ep$;  this implies that
\begin{equation*}
 \sup_{x \in \Gamma} \mathcal{N}^{\ep, \alpha}_n (x) \leq \mathcal{N} (\tilde{\gamma}^\ep_i) = \left( \frac{(1 + 4 \lambda \ep^\alpha B) T }{ n \ep } + 1 \right)^2.
 \end{equation*}
 
Hence, by the Poisson field's translation invariance,
$$ P^\ep (\{ \sup_{t \leq T} \geq \lambda \}) = \left( \frac{(1 + 4 \lambda \ep^\alpha B) T}{ n \ep } + 1 \right)^2 P^\ep (\mathcal{N}(\tilde{\gamma}^\ep_i) \geq \lambda).$$
Markov's inequality with an exponential finishes the proof.
\end{proof}

\subsection{Bounding the crossing times}

We would like to control, again for $\alpha \in (0,1/2)$, the time that the particle takes to cross the ball $S^\ep_n (x^{\ep,\alpha}_s)$, compared to a trajectory crossing the ball with constant velocity. Define:
$$U_\lambda :=  \{ \sup_{s \leq T} \mathcal{N}^{\ep,\alpha}_{n+1} (x^{\ep, \alpha}_s) \leq \lambda \}  \cap \{ \sup_{s \leq T} \abs{v^{\ep, \alpha}_s} \geq 1/2 \}.$$  We will call the first set in this intersection $\Lambda^\ep_{n+1}$, and we write $B' := \sup \abs{ F(x) }.$

\begin{lemma}\label{CTL}
If $t^\ep_n(s)$ and $\tilde{t}^\ep_n(s)$ are the first entrance and exit times of the trajectory with respect to the ball $S^\ep_n (x^\ep_s)$, then for trajectories in $U_\lambda$ and for $\ep <  (16 n \lambda B') ^{ -1/\alpha }$, we have:
\begin{gather}
 \tilde{t}^\ep_n (s) - t^\ep_n (s) \leq 4 n \ep^\alpha, \\
  \sup_{s \leq T} \sup_{ t^\ep_n (s) \leq t \leq \tilde{t}^\ep_n (s) } \abs{ v^\ep_\alpha (t) - v^\ep_\alpha (t^\ep_n (s)) } = O( \ep^{\alpha} ).
  \end{gather}
\end{lemma}

\begin{proof}
If a particle enters $S^\ep_n (x^{\ep,\alpha}_s)$ with constant velocity $v^{\ep,\alpha}_1 =  v^\ep_\alpha (t^\ep_n (s)) / 2$, position $x^\ep_\alpha (t^\ep_n (s))$, at time $t^\ep_n (s)$, then its exit time $\hat{t}^\ep_n (s)$ will satisfy 
$$ \hat{t}^\ep_n (s) - t^\ep_n (s) \leq 2n \ep \abs{ v^\ep_1 }^{-1} \leq 4n\ep^\alpha. $$

By the point process description, \eqref{ptprocessdiff}, we have for $t \leq \tilde{t}^\ep_n (s)$:
$$ \abs { v^\ep_\alpha (t) - v^\ep_\alpha (t^\ep_n (s)) } \leq \sum_{r \in \omega \cap S^\ep_{n+1} (x^{\ep,\alpha}_s) } \biggl\lvert  \int^t_{t^\ep_n (s)} F^\ep_\alpha (x^\ep_\alpha(t') - r) dt' \biggr\rvert.$$

Putting these two facts together, we have for $t < \hat{t}^\ep_n (s)$:
\begin{equation*}
\begin{split} \abs { v^\ep_\alpha (t) - v^{\ep,\alpha}_1 } & \leq \mathcal{N}^{\ep,\alpha}_{n+1} (x^{\ep,\alpha}_s) \ep^{-1} \ep^{2\alpha - 1} B' 4 n \ep^\alpha \\
& = 4 n B' \mathcal{N}^{\ep,\alpha}_{n+1} (x^{\ep,\alpha}_s) \ep^{\alpha}.
\end{split}
\end{equation*}

So on $U_\lambda$,  and for $\ep < (16 n \lambda B') ^{ -1/\alpha }$, we have:
$$ v^\ep_\alpha (t) \cdot \frac{v^{\ep,\alpha}_1}{\abs{ v^{\ep,\alpha}_1 }} \geq \abs{ v^{\ep,\alpha}_1 } - \mathcal{N}^{\ep,\alpha}_{n+1} (x^{\ep,\alpha}_s) 4 n B' \ep^{\alpha} > \abs{ v^{\ep,\alpha}_1 } / 2.$$
\end{proof}

\subsection{Tightness}

We now turn to showing tightness for $\tilde{\nu}^\ep_\alpha$, the family of measures induced by the cut-off processes $v^{\ep,\alpha}_t = v^\ep_\alpha (t \wedge \tau^\ep).$  

\begin{lemma}
The family $\tilde{\nu}^\ep_\alpha$ is tight in $C([0,T]; \mathbb{R}^2 \times \mathbb{R}^2)$: that is, for all $\gamma, \eta > 0$, there exists $\delta$ such that for $\ep$ small enough, 
$$ P^\ep \Biggl\{ \sup_{\substack{s, t \leq T \\ \abs{t-s} < \delta}} \abs{v^{\ep,\alpha}_t - v^{\ep,\alpha}_s} > \gamma \Biggr\} <\eta.$$
\end{lemma}

\begin{proof}
We use the point process description \eqref{ptprocessdiff} to write:
\begin{equation}\label{ABCDdecomp}
 v^{\ep,\alpha}_t - v^{\ep,\alpha}_s = A + B + C + D,
 \end{equation}
where
\begin{gather} 
A:= \int_S \int_{s \wedge \tau^\ep}^{t \wedge \tau^\ep} N^\ep_\alpha (d\sigma, du) \int_u^{t^\ep_{u, \sigma}} F^\ep_\alpha (x^{\ep,\alpha}_{u, \sigma} (t')) dt',\\
B:= \int_S \int_{s \wedge \tau^\ep}^{t \wedge \tau^\ep} N^\ep_\alpha (d\sigma, du) \int_{t^\ep_{u, \sigma}}^{t \wedge \tau^\ep} F^\ep_\alpha(x^{\ep,\alpha}_{u, \sigma}(t')) dt',\\
C:= \int_S \int_0^{t \wedge \tau^\ep} N^\ep_\alpha (d\sigma, du) \int_{s \wedge \tau^\ep}^{t \wedge \tau^\ep} F^\ep_\alpha(x^{\ep,\alpha}_{u, \sigma}(t')) dt',\\
D:= \Delta v^{\ep,\alpha}_0(t) - \Delta v^{\ep,\alpha}_0(s).
\end{gather}

Here we have defined $t^\ep_{u, \sigma} := \inf \{ t > u: x^{\ep,\alpha}_{u, \sigma}(t) \notin B(x^{\ep,\alpha}_u - \sigma \ep, \ep) \} \wedge \tau^\ep,$ i.e., the particle's first exit time after $u$.  Terms $B, C$, and $D$ concern scatterers that the particle encounters at self-crossings, at time $s \wedge \tau^\ep$, and at time $0$, respectively.  By the overlap and crossing-time lemmas, \ref{OL} and \ref{CTL}, we can bound these terms:
\begin{equation}\label{supBCD}
\sup_{s,t\leq T} \abs{ B + C + D } \leq 4KB'n \lambda \ep^{\alpha}.
\end{equation}
If $\lambda = \ep^{-\beta}$, with $\beta < 2\alpha$, then these terms are negligible in the limit $\ep \to 0$.

To handle $A$, on the other hand, requires the martingale decomposition, \eqref{MGdecomp}.  However, we run into a problem trying to implement this decomposition, because the function 
$f(\sigma, u) := \int_u^{t^\ep_{u, \sigma}} F^\ep_\alpha (x^{\ep,\alpha}_{u, \sigma}(t')) dt'$, is not adapted to the family of sigma algebras,  $\mathcal{F}_t$, generated by $N^\ep_\alpha$--instead, it's anticipative. To replace $f$ by an adapted function, we do a Taylor expansion of $F^\ep_\alpha$ around the line $\hat{x}^{\ep,\alpha}_{u, \sigma} (t) = v^{\ep,\alpha}_u(t-u) + \sigma \ep$, for $t \in [u, \hat{t}^\ep_{u, \sigma}],$ where $\hat{t}^\ep_{u, \sigma}$ is the first re-entrance time after $u$:
\begin{equation*}
 \hat{t}^\ep_{u, \sigma} := \inf \{ t > u: x^{\ep,\alpha}_{u, \sigma}(t) \in B(x^{\ep,\alpha}_u - \sigma \ep, \ep) \} \wedge \tau^\ep.
 \end{equation*}

Then the Taylor expansion is: $\quad I  \; \quad \quad \quad  \quad \quad \quad \quad \quad \quad II$
\begin{equation}\label{TaylorExp}
\int_u^{t^\ep_{u, \sigma}} F^\ep_\alpha(x^{\ep,\alpha}_{u, \sigma}(t'))dt' = \overbrace{ \int_u^{\hat{t}^\ep_{u, \sigma}} F^\ep_\alpha (\hat{x}^{\ep,\alpha}_{u, \sigma} (t')) dt' } + \overbrace{ \int_{\hat{t}^\ep_{u, \sigma}}^{t^\ep_{u, \sigma}} F^\ep_\alpha (x^{\ep,\alpha}_{u, \sigma} (t'))dt'}  + \; III + IV,
\end{equation}
where $III$ is the linear term, and $IV$ is the remainder, with mean value $y^\ep(u,\sigma,t')$:
$$III := \int_u^{\hat{t}^\ep_{u, \sigma}} (x^{\ep,\alpha}_{u, \sigma}(t') - \hat{x}^{\ep,\alpha}_{u, \sigma}(t')) \cdot \nabla F^\ep_\alpha(\hat{t}^\ep_{u, \sigma}(t'))dt',$$ 
$$IV := \int_u^{\hat{t}^\ep_{u, \sigma}} \frac{1}{2} ((x^{\ep,\alpha}_{u, \sigma}(t') \hat{x}^{\ep,\alpha}_{u, \sigma}(t')) \cdot \nabla)^2 F^\ep_\alpha (y^\ep(u,\sigma,t')) dt'.$$ 

We observe that term $I$ is $\mathcal{F}_u$-adapted and continuous in $u$.  A useful estimate from the crossing time lemma and $\sup_{u\leq \tau^\ep} \abs{ \hat{t}^\ep_{u, \sigma} - u } \leq 4\ep$ is:
$$\sup_{u \leq t' \leq \hat{t}^\ep_{u, \sigma}} \abs { x^{\ep,\alpha}_{u, \sigma} (t') - \hat{x}^{\ep,\alpha}_{u, \sigma} (t') } \leq C \lambda \ep^{1 + \alpha}.$$

The other terms are still anticipative, so they need to be analyzed further.  For terms $II$ and $IV$, it suffices (by sending $\lambda \to \infty$ in the overlap lemma) to show that they decay asymptotically on the set 
$$\Lambda^\ep_n := \{ \sup_{s \leq T} \mathcal{N}^{\ep,\alpha}_n (x^{\ep,\alpha}_s) \leq \lambda \} .$$

In order to bound $II$ and $IV$, we use the crossing time lemma, \ref{CTL}, to show that on $\Lambda^\ep_n$, both of the following are of order $ \lambda^2 \ep^{1 + \alpha}$:
\begin{gather*} \sup_{s \leq t \leq T} \int_{s \wedge \tau^\ep}^{t \wedge \tau^\ep} \int_S N^\ep_\alpha(d\sigma, du)  \int_{\hat{t}^\ep_{u, \sigma}}^{t^\ep_{u, \sigma}} F^\ep_\alpha (x^{\ep,\alpha}_{u, \sigma} (t'))dt',\\
 \sup_{s \leq t \leq T}  \int_{s \wedge \tau^\ep}^{t \wedge \tau^\ep}   \int_S  N^\ep_\alpha(d\sigma, du)   \int_u^{\hat{t}^\ep_{u, \sigma}}  \frac{1}{2} ((x^{\ep,\alpha}_{u, \sigma}(t') \hat{x}^{\ep,\alpha}_{u, \sigma}(t'))  \cdot  \nabla)^2 F^\ep_\alpha (y^\ep(u,\sigma,t')) dt'
 \end{gather*}
These terms go to zero as $\ep \to 0$, for $\lambda = \ep^{-\beta}$, with $\beta < 1/3$. Hence, for term $IV$, we can estimate:
\begin{equation*}
\begin{split} 
\sup & \left \lvert \int_{s \wedge \tau^\ep}^{t \wedge \tau^\ep} \int_S N^\ep_\alpha (d\sigma, du)   \int_u^{\hat{t}^\ep_{u, \sigma}} \frac{1}{2} ([x^{\ep,\alpha}_{u, \sigma}(t') - \hat{x}^{\ep,\alpha}_{u, \sigma}(t')] \cdot \nabla)^2 F^\ep_\alpha ( y^\ep (u,\sigma,t'))dt'  \right \rvert \\
	& \leq \left( \int_S \int_0^{T \wedge \tau^\ep} N^\ep_\alpha (d\sigma, du) \right) C \lambda^2 \ep^{3\alpha} \sup_{e_1, e_2} \abs{ e_1 \cdot \nabla e_2 \cdot \nabla F }. 
\end{split}
\end{equation*}
Choosing $\lambda = \ep^{-\beta}$, with $\beta < 3\alpha$, implies that term $IV$ is negligible as $\ep \to 0$.

As for term $II$, observe that on $\Lambda^\ep_n$, for all $u$,
\begin{equation}\label{est1}
\begin{split} 
\abs{ x^{\ep,\alpha}_{u, \sigma} (t^\ep_{u, \sigma}) - \hat{x}^{\ep,\alpha}_{u, \sigma} (\hat{t}^\ep_{u, \sigma})) } & = \abs{ x^{\ep,\alpha}_t - x^{\ep,\alpha}_u + \sigma \ep - v^{\ep,\alpha}_u(\hat{t}^\ep_{u, \sigma} - u) - \sigma \ep } \\
	&  = \left \lvert \int_u^t v^{\ep,\alpha}_s v^{\ep,\alpha}_u(\hat{t} -u) \frac{1}{t-u} ds  \right \rvert  \leq  C \lambda \ep^{1+\alpha}.
\end{split}
\end{equation}
From the crossing time lemma, we have $t^\ep_{u, \sigma} - u < C\ep$, and hence
$$ \sup_{u \leq t \leq t^\ep_{u, \sigma}} \abs{ v^{\ep,\alpha}_t - v^{\ep,\alpha}_u } < C' \lambda \ep^{\alpha}.$$
Then $\abs{ t^\ep_{u, \sigma} - \hat{t}^\ep_{u, \sigma} } < C'' \lambda \ep^{1+\alpha} $, so we can deduce a version of $\eqref{est1}$, uniformly in $t$:
\begin{equation*}
\sup_{\hat{t}^\ep_{u, \sigma} \leq t \leq t^\ep_{u, \sigma}} \abs{ x^{\ep,\alpha}_{u, \sigma} (t) - \hat{x}^{\ep,\alpha}_{u, \sigma} (\hat{t}^\ep_{u, \sigma})) } \leq  C \lambda \ep^{1+\alpha}.
\end{equation*}
Then since
$$\sup_{ \abs{ y - \partial S^\ep } < C'' \lambda \ep^{2+2\alpha} } \abs{ F^\ep_\alpha (y) } = B' C'' \lambda,$$
\begin{equation*}
\sup \left \lvert \int_{s \wedge \tau^\ep}^{t \wedge \tau^\ep} \int_S N^\ep_\alpha (d\sigma, du)   \int_u^{\hat{t}^\ep_{u, \sigma}} F^\ep_\alpha (x^{\ep,\alpha}_{u, \sigma} (t'))dt'  \right \rvert \leq N^\ep_\alpha(T) C''^2 B' \lambda^2 \ep^{1+\alpha}.
\end{equation*}

Now to estimate $III$, decompose $v^{\ep,\alpha}_t$ much like \cite{DGL}:
\begin{equation}\label{ypmoDecomp}
v^{\ep,\alpha}_t \sim y^\ep_+(t) + y^\ep_-(t) + y^\ep_0(t).
\end{equation}
The three processes $y^\ep_{+,-,0}$ are defined as follows:
\begin{gather}\label{ypmdiff}
 y^\ep_0 (t) := \int_S \int_0^{t \wedge \tau^\ep} N^\ep_\alpha (d\sigma, du) \int_u^{\hat{t}^\ep_{u, \sigma}} F^\ep_\alpha (\hat{t}^\ep_{u, \sigma} (t')) dt', \\
 y^\ep_{\pm} (t) - y^\ep_{\pm}(s) := \int_S \int_{s \wedge \tau^\ep}^{t \wedge \tau^\ep} N^\ep_\alpha (d\sigma, du) \Delta v^{\ep,\alpha}_{\pm}(u,\sigma).
 \end{gather}
Here we abbreviate $d\mathbf{t} = dt''' dt'' dt'$ and define:
\begin{equation*}
\begin{split}
\Delta & v^{\ep,\alpha}_+ (u, \sigma) := \int_u^{\hat{t}^\ep_{u, \sigma}} \int_{u}^{t'} \int_{u}^{t''} \int_S \hat{\rho}^\ep_\alpha(d\sigma', du') \int_{u'}^{t''} F^\ep_\alpha (\hat x^{\ep,\alpha}_{u', \sigma'} (t''')) \cdot \nabla F^\ep_\alpha (\hat{x}^{\ep,\alpha}_{u, \sigma} (t')) d\mathbf{t} \\
& = \int_{s \wedge \tau^\ep}^{u} N^\ep_\alpha (d\sigma', du') \chi(\hat{t}^\ep_{u, \sigma} \geq u) \int_{u'}^{\hat{t}^\ep_{u, \sigma}} \int_{u}^{t'} \int_{u}^{t''} F^\ep_\alpha (\hat x^{\ep,\alpha}_{u', \sigma'} (t''')) \cdot \nabla F^\ep_\alpha (\hat{x}^{\ep,\alpha}_{u, \sigma} (t')) d\mathbf{t};
\end{split}
\end{equation*}
\begin{equation}\label{deltavem}
\begin{split}
\Delta v^{\ep,\alpha}_- (u, \sigma) & := \! \int_u^{\hat{t}^\ep_{u, \sigma}} \! \int_u^{t'} \! \int_{u - 4 \ep^2}^u \int_S \! N^\ep_\alpha (d\sigma', du') \int_{u'}^{t''} \! F^\ep_\alpha (\hat x^{\ep,\alpha}_{u', \sigma'} (t''')) \cdot \nabla F^\ep_\alpha (\hat{x}^{\ep,\alpha}_{u, \sigma} (t')) d\mathbf{t} \\ 
 & = \Delta \tilde{v}^\ep _- (u, \sigma) + \int_u^{\hat{t}^\ep_{u, \sigma}}  \int^{t'}_u
 \Delta N^\ep_\alpha (u) \int_u^{t''} F^\ep_\alpha (\hat x^{\ep,\alpha}_{u', \sigma'} (t''')) \cdot \nabla F^\ep_\alpha (\hat{x}^{\ep,\alpha}_{u, \sigma} (t')) d\mathbf{t}.
\end{split}
\end{equation}
Noting that $\Delta v^{\ep,\alpha}_+(u,\sigma)$ is already adapted and left-continuous, we have further split $\Delta v^{\ep,\alpha}_-$ into its adapted left-continuous part, $\Delta \tilde v^{\ep,\alpha}_-$, and a jump part, with $\Delta N^\ep_\alpha (u) := N^\ep_\alpha(S, [0,u]) - N^\ep_\alpha (S,[0,u)).$

It remains to show that:
\begin{equation}\label{ktight}
\begin{split}
&\text{(i) There is a positive constant } C \text{ such that for every } \ep \ll 1, \\
& \quad \quad \quad  \mathbb{E}^\ep ( | y^\ep_{\pm} (t) - y^\ep_{\pm}(s) |^2) \leq C |t - s|^2, \text{ for } |t-s| < \ep; \\
&\text{(ii) For every } \gamma, \eta > 0 \text{ and } \ep \ll 1, \\
&  \quad \quad \quad  P^\ep(\{ \sup_{\substack{ s,t \leq T \\  |s-t|<\ep^2}} | y^\ep_{\pm} (t) - y^\ep_{\pm}(s) | > \gamma \} ) < \eta.
\end{split}
\end{equation}
Similar to \cite{DGL}, the proof uses the martingale splitting developed earlier, as well as the Burkholder-Davis-Gundy inequalities. For example, to obtain the first inequality of $\eqref{ktight}$ for $y^\ep_-$, we use the martingale-compensator splitting:
\begin{equation}\label{tightdecomp}
\begin{split}
\mathbb{E}^\ep \big( \abs{ y^\ep_- (t) - y^\ep_- (s) } \big) & = \mathbb{E}^\ep \left[ \iint N^\ep_\alpha \Delta \tilde{v} ^\ep _- + \iint N^\ep_\alpha \iiint F^\ep_\alpha \cdot \nabla F^\ep_\alpha \right] ^2 \\
& \leq 4 \left( \mathbb{E}^\ep \left[ \left( \iint \rho^\ep_\alpha \Delta \tilde{v} ^\ep_- \right)^2 \right] +  \mathbb{E}^\ep \left[ \left( \iint M^\ep_\alpha \Delta \tilde{v} ^\ep_- \right)^2 \right] \right. \\
& \left. \quad \quad  \quad \quad + \,  \mathbb{E}^\ep \left[ \left( \iint (\rho^\ep_\alpha + M^\ep_\alpha) \iiint F^\ep_\alpha \cdot \nabla F^\ep \right)^2 \right] \right).
\end{split}
\end{equation}
We also need a bound on the Poisson rate: $\rho^\ep_\alpha (d\sigma,du) \leq \frac{3}{2} \rho^\ep_\alpha \ep d\sigma du = \frac{3}{2} \rho \ep^{-2\alpha} d\sigma du.$
Then $N^\ep_\alpha$ is dominated by $\tilde{N}^\ep_\alpha$, which is Poisson with density $\frac{3}{2} \pi \rho \ep^{-4\alpha} d\sigma du.$
And for the first term on the right-hand side of $\eqref{tightdecomp}$:
\begin{equation*}
 \mathbb{E}^\ep \left[ \left( \iint \rho \Delta v^{\ep,\alpha}_- \right) ^2 \right] \leq \mathbb{E} \left[ \Big( \iint \rho \Big)^2 (\Delta v^{\ep,\alpha}_-)^2 \right]
 \leq C[\rho \ep^{-2\alpha}]^2 \abs{ t-s }^2 \ep^{4\alpha}  = C \abs{ t-s }^2.
\end{equation*}

Then use the quadratic variation formula and continuity of $\rho^\ep$ to bound the second term by the first:
\begin{equation*}
\mathbb{E}^\ep \left[ \left( \iint M^\ep_\alpha \Delta \tilde{v} ^\ep_- \right)^2 \right] = \mathbb{E}^\ep \left[ \left( \iint \rho^\ep \Delta \tilde{v} ^\ep_- \right)^2 \right] = \mathbb{E}^\ep \left[ \left( \iint \rho^\ep \Delta v ^\ep_- \right)^2 \right]. 
\end{equation*}

Thus the first two terms of $\eqref{tightdecomp}$ are bounded by a multiple of $\abs{ t - s}^2$. The other terms can be handled similarly, as can the inequality for $y^\ep_+$, since 
\begin{equation*}
\sup_{u\leq \tau^\ep} \abs{ \hat{t}^\ep_{u, \sigma} - u } \leq 4\ep.
\end{equation*}

Next we handle tightness for $y^\ep_0$, showing that for some positive constant $C$,
\begin{equation*}
\mathbb{E}^\ep(\abs{y^\ep_0 (t) - y^\ep_0(s) }^{2+\gamma}) \leq C' \abs{t-s}^{1+\gamma/2}.
\end{equation*}
 
Let $K$ be the number of self-crossings in the trajectory, and let $T^\ep$ be the tube around the trajectory up to the stopping time $\tau^\ep$. At the $i$-th self-crossing, let $T^\ep_i$ and $\hat{T}^\ep_i$ be the times that the particle enters and exits the tube. Set $U := [0, T^\ep_1] \cup [\hat{T}^\ep_1, T^\ep_2] \cup \cdots \cup [\hat{T}^\ep_{K-1}, T^\ep_K]$, and $Poi^\ep_\alpha $ is a Poisson point measure with intensity $\hat{\rho}^\ep_\alpha (d \sigma, du) := - \rho \ep^{-2\alpha} v^{\ep,\alpha}_u \cdot d \sigma du \vee 0.$ Then define:
\begin{equation*}
\hat{N}^\ep_\alpha (d \sigma, du) := \begin{cases}
N^\ep_\alpha (d \sigma, du)& \text{ on } U, \\
Poi^\ep_\alpha  (d \sigma, du)& \text{ off } U.
\end{cases}
\end{equation*}
Then \ref{CTL} implies that there is a number $n = n (\phi, a)$ independent of $\ep$ and such that 
\begin{equation*}
\{ \sup_{s \leq T} \mathcal{N}^{\ep,\alpha}_{n(\phi, a)} (s) \leq \lambda \}.
\end{equation*}

The remainder of the proof, showing tightness for $\hat{y}^\ep_0 (t) := \iint \hat{N}^\ep_\alpha \int F^\ep_\alpha (x^{\ep,\alpha}_{u, \sigma})$ is a straightforward generalization of arguments in \cite{DGL}. \end{proof}

\section{Characterization of the limit}

To identify the limiting process as the velocity diffusion associated to the linear Landau equation, we use the Stroock-Varadhan martingale formulation, thereby showing that the tight family of processes converges to the velocity diffusion, we show that two quantities involving the process and the infinitesimal generator of the diffusion are martingales; the desired result will follow by Levy's lemma.

The infinitesimal generator of the diffusion process is:
\begin{equation}\label{infgenerator}
L = \pi \rho \nabla_p \int d^2 k (k \otimes k) \delta(k \cdot p) \abs{\hat{V}(\abs{k})}^2 \cdot \nabla_p.
\end{equation}
It can be expressed more simply in polar coordinates, $p = (r, \theta)$, as the Laplace-Beltrami operator: \begin{equation}
L = \zeta \partial^2 / \partial \theta^2.
\end{equation}

\begin{lemma}
Let $s,t \in [0,T]$, $f$ be a smooth test function, and $\phi_s$ be a smooth and bounded test function that depends only on $p(u), u \leq s$, where $p \in C([0,T]^2; \mathbb{R}^2 \times \mathbb{R}^2)$. Also write $\phi^\ep_s = \phi_s (v^{\ep,\alpha}_u).$  Then we have
\begin{equation}\label{mg1}
\lim_{\ep \to 0} \mathbb{E}^\ep_\alpha \left[ f \left( v^{\ep,\alpha}_t - v^{\ep,\alpha}_s - \int_{s \wedge \tau^\ep}^{t \wedge \tau^\ep} L v^{\ep,\alpha}_u du \right) \phi_s(v^\ep) \right] = 0;
\end{equation}
\end{lemma}

\begin{proof}
It suffices to prove \eqref{mg1} for linear and quadratic functions $f$; we restrict our attention here to the linear case. Using the decomposition \eqref{ABCDdecomp} we can see that the $B$, $C$, and $D$ terms are negligible in the limit, as follows. We  cut their expectation into two pieces, $S:= \{ \sup_{s \leq T} \mathcal{N}^{\ep,\alpha}_{n(\phi,a)} (s) \leq \lambda  \}$ and $S^c$; and on the first term we use the bound on the supremum of $B+C+D$ from \eqref{supBCD}, and on the second term we use the overlap lemma, \eqref{OL}:
\begin{equation*}
\begin{split}
\mathbb{E}^\ep[ (B & +C+D) \phi^\ep_s ]  \leq C\lambda \ep^\alpha + C' T \ep^{-1} \mathbb{E}^\ep[ \sup_{s \leq T} \mathcal{N}^\ep_{n(\phi,a)} (s) \chi\{ \sup_{s \leq T} \mathcal{N}^\ep_{n(\phi,a)} (s) > \lambda \} ] \\
& \leq C\lambda \ep^\alpha + C' T \ep^{-1} \int_\lambda^\infty P^\ep \{ \sup_{s \leq T} \mathcal{N}^\ep_{n(\phi,a)} (s) > \lambda' \} d\lambda' \\
& \leq C\lambda \ep^\alpha + C' T \ep^{-1} \int_\lambda^\infty \left( \frac{1+ 4\lambda'\ep^\alpha BT}{n\ep} + 1 \right)^2 \exp{ 32n^2 \rho \ep^{-2\alpha - 1} - \lambda' } d\lambda' \\
& \leq C\lambda \ep^\alpha + C' T \ep^{-1} h(\lambda, \ep,T) e^{-\lambda} \\
\end{split}
\end{equation*} 
Here, $h$ involves terms like $\lambda^2 T^2 \ep^{2\alpha - 2}$, and $\lambda$ can be chosen to be a small enough (for example, $\ep^{-\alpha + 3/2 + \beta} T^{-1 - \beta}$, with  $\beta$ small and positive) to get the right-hand side to vanish with $\ep$.
Hence only the term $A$ matters:
\begin{equation*}
\mathbb{E}^\ep [ (v^{\ep,\alpha}_t - v^{\ep,\alpha}_s) \phi_s(v^\ep) ]  \sim  \mathbb{E}^\ep \left[ \left( \int_S \int_{s \wedge \tau^\ep}^{t \wedge \tau^\ep} N^\ep_\alpha(d\sigma, du) \int_u^{t^\ep_{u, \sigma}} F^\ep(x^{\ep,\alpha}_{u, \sigma}(t')) dt' \right) \phi^\ep_s \right].
\end{equation*}

Next recall the decomposition of \eqref{ypmoDecomp}, $v^{\ep,\alpha}_t \sim y^\ep_+(t) + y^\ep_-(t) + y^\ep_0(t)$.
We define $\hat{y}^\ep_0(t)$ and $\hat{y}^\ep_{\pm}(t)$ by replacing $N^\ep_\alpha$ by $\hat{N}^\ep_\alpha$, the point process that ignores self-crossings, in their formulas; this results in negligible errors. Then we can combine this with the Taylor expansion in \eqref{TaylorExp}, and use the martingale property of $\hat{y}^\ep_0$ to arrive at:
\begin{equation}
\mathbb{E}^\ep [ (v^{\ep,\alpha}_t - v^{\ep,\alpha}_s) \phi^\ep_s]  \sim \mathbb{E}^\ep [ (\hat{y}^\ep_-(t) - \hat{y}^\ep_-(s) + \hat{y}^\ep_+(t) - \hat{y}^\ep_+(s) ) \phi^\ep_s ].
\end{equation}
Then we can use the formulas from the previous section, $\eqref{ypmdiff}$ and \eqref{deltavem}, to reduce the proof of the lemma to the following two limits:
\begin{equation}\label{twolimits}
\begin{split}
\lim_{\ep \to 0} \mathbb{E}^\ep & \left[ \left( \int_S \int_{s \wedge \tau^\ep}^{t \wedge \tau^\ep} \hat{N}^\ep_\alpha (d\sigma, du) [\Delta \hat{\tilde{v}}^\ep_- (u,\sigma) + \Delta \hat{v}^\ep_+(u,\sigma) ] \right) \phi^\ep_s \right] = 0 \\
\lim_{\ep \to 0} \mathbb{E}^\ep & \left[ \left( \int_S \int_{s \wedge \tau^\ep}^{t \wedge \tau^\ep} \hat{N}^\ep_\alpha (d\sigma, du) \int_u^{\hat{t}^\ep_{u, \sigma}} \int_u^{t'} \int_u^{t''} F^\ep(\hat{x}^{\ep,\alpha}_{u, \sigma}(t''')) \cdot \nabla F^\ep_\alpha (\hat{x}^{\ep,\alpha}_{u, \sigma}(t')) \right) \phi^\ep_s \right] \\
& = \lim_{\ep \to 0} \mathbb{E}^\ep \left[ \left( \int_{s \wedge \tau^\ep}^{t \wedge \tau^\ep} L v^{\ep,\alpha}_u du \right)  \phi^\ep_s \right].
\end{split}
\end{equation}

The first of these two limits can be handled by manipulations similar to \cite{DGL}, except with $\hat{\rho}^\ep_\alpha (d\sigma', du') = - \rho \ep^{-2\alpha} v^{\ep,\alpha}_u \cdot d\sigma du \vee 0$ and $\hat{x}^{\ep,\alpha}_{u', \sigma', u} (t) := v^{\ep,\alpha}_u (t - u') + \sigma' \ep.$

For the second of the two limits in \eqref{twolimits}, we examine the part without the infinitesimal operator. Using the martingale decomposition of $\hat{N}^\ep_\alpha$ and the definitions of $\hat{\rho}$, $F^\ep_\alpha$, and $\nabla F^\ep_\alpha$, we arrive at:
\begin{equation*}
\begin{split}
& \mathbb{E}^\ep \left[ \left( \int_S \int_{s \wedge \tau^\ep}^{t \wedge \tau^\ep} \hat{N}^\ep_\alpha (d\sigma, du) \int_u^{\hat{t}^\ep_{u, \sigma}} \int_u^{t'} \int_u^{t''} F^\ep_\alpha(\hat{x}^{\ep,\alpha}_{u, \sigma}(t''')) \cdot \nabla F^\ep_\alpha (\hat{x}^{\ep,\alpha}_{u, \sigma}(t')) \right) \phi^\ep_s \right]  \\
& \sim \mathbb{E}^\ep \left[ \left( \int_S \int_{s \wedge \tau^\ep}^{t \wedge \tau^\ep} \hat{\rho}^\ep_\alpha (d\sigma, du) \int_u^{\hat{t}^\ep_{u, \sigma}} \int_u^{t'} \int_u^{t''} F^\ep_\alpha(\hat{x}^{\ep,\alpha}_{u, \sigma}(t''')) \cdot \nabla F^\ep_\alpha (\hat{x}^{\ep,\alpha}_{u, \sigma}(t')) \right) \phi^\ep_s \right] \\
& = \mathbb{E}^\ep \left[ \left( \int_{s \wedge \tau^\ep}^{t \wedge \tau^\ep} -\rho \ep^{-2 \alpha} du \int_{ v^{\ep,\alpha}_u \cdot \sigma } v^{\ep,\alpha}_u d\sigma \int_u^{\hat{t}^\ep_{u, \sigma}} \int_u^{t'} \int_u^{t''} \ep^{\alpha - 1} F \left( \frac{v^{\ep,\alpha}_u(t''' - u)}{\ep} + \sigma \right) \right. \right. \\
& \quad \quad \quad \quad \quad \quad \quad \quad \quad \left. \left. \cdot \ep^{\alpha-2} \nabla F \left(  \frac{v^{\ep,\alpha}_u(t' - u)}{\ep} + \sigma  \right) \right) \phi^\ep_s \right] 
\end{split}
\end{equation*}
Now we manipulate the inner integrals to make them look like the desired operator:
\begin{equation}\label{DGLdiffusion2}
\begin{split}
- \int_{s \wedge \tau^\ep}^{t \wedge \tau^\ep}  \rho \ep^{-2 \alpha} & \int_{v_u \sigma \geq 0}  \int_{-\infty}^{\infty}  \int_{-\infty}^{t'} (t' - t'') F(v^{\ep,\alpha}_u t'' + \sigma) \cdot \nabla F(v^{\ep,\alpha}_u t' + \sigma) dt'' dt' d\sigma du \\
& = \int_{s \wedge \tau^\ep}^{t \wedge \tau^\ep}  \rho \ep^{-2 \alpha}  \int  \int_{-\infty}^{0} \tau F(r+ v^{\ep,\alpha}_u \tau) \cdot \nabla F(r)  d\tau d^2 r du \\
& = - \int_{s \wedge \tau^\ep}^{t \wedge \tau^\ep} \frac{1}{2} \rho \ep^{-2 \alpha}  
\int \int_{-\infty}^{\infty}  [\nabla_\xi \cdot F(r + \xi \tau) F(r) ]  \upharpoonright_{ \xi = v^{\ep,\alpha}_u } d\tau d^2 r du \\
& = - \frac{1}{2} \rho \int_{s \wedge \tau^\ep}^{t \wedge \tau^\ep} \ep^{-2 \alpha} \nabla_\xi \cdot \int \int_{-\infty}^{\infty} e^{ik \cdot \xi \tau} (k \otimes k) \abs{\hat{V}(\abs{k}) }^2 \upharpoonright_{ \xi = v^{\ep,\alpha}_u } d\tau d^2 k du \\
& = - \int_{s \wedge \tau^\ep}^{t \wedge \tau^\ep}  \rho \ep^{-2 \alpha} \nabla_\xi \cdot \int \pi \delta(k \cdot \xi) (k \otimes k) \abs{\hat{V}(\abs{k}) }^2 \upharpoonright_{ \xi = v^{\ep,\alpha}_u } d^2 k  du \\
& = - \int_{s \wedge \tau^\ep}^{t \wedge \tau^\ep}  \ep^{-2 \alpha} L v^{\ep,\alpha}_u du.
\end{split}
\end{equation}
\end{proof}

\section{Convergence in expectation}

We will show that, in particular, for an initial distribution $f_0$,
\begin{equation}
f^\ep_\alpha(t,x,v) = \mathbb{E} [ f_0(\Phi^{t}_{\alpha, \omega, \ep} (x,v)) ] \xrightarrow{\ep \to 0} h(t,x,v),
\end{equation}
where $h$ is the solution of the linear Landau equation, \eqref{LinearLandau}. Although this is a weaker mode of convergence than that just proved, the following argument illustrates the key intuition behind the result for the entire range of $\alpha$: Although there are significant numbers of clusters of obstacles, their total influence is actually negligible. The initial set-up is similar to \cite{DR} but differs after the first two steps.  

First, observe that $f^\ep_\alpha(t,x,v)$ can be written as (using $\abs{B(x)}$ to denote the measure of the ball of radius $t$ around the initial position $x$):
\begin{equation*}
f^\ep_\alpha(t,x,v) = e^{-\rho^\ep_\alpha \abs{B(x)}} \sum_{N \geq 0} \frac{(\rho^\ep_\alpha)^N}{N!} \int_{B(x)} \cdots \int_{B(x)} f_0(\Phi^{t}_{\alpha, \omega, \ep} (x,v)) dr_1 \dots dr_N.
\end{equation*}
Define a cut-off $\chi_1$, killing configurations that have an obstacle at the initial position:
\begin{equation*}
\chi_1(\omega) := \chi( \{ \omega = \{r_i\}_{i = 1}^N : \forall i = 1, \dots, N, \abs{ x - r_i } > \ep \} ).
\end{equation*}
Making this cut-off introduces an asymptotically vanishing error. That is, there exists a function $\phi_1(\ep) \to 0$ as $\ep \to 0$ such that $f^\ep_\alpha$  can be written as:
\begin{equation*}
e^{- \rho^\ep_\alpha \abs{B(x)}} \sum_{N \geq 0} \frac{(\rho^\ep_\alpha)^N}{N!} \int_{B(x)} \cdots \int_{B(x)} \chi_1(\omega) f_0(\Phi^{t}_{\alpha, \omega, \ep} (x,v)) dr_1 \dots dr_N + \phi_1(\ep).
\end{equation*}
Next we define another cut-off, $\chi_2$, killing configurations with obstacles that are not encountered by the light particle and thus have no effect on the trajectory:
\begin{equation*}
\chi_2(\omega) := \chi( \{ \omega = \{ r_i \}_{i=1}^M : \forall i = 1, \dots, M, r_i \in \mathcal{T}(t) \} ).
\end{equation*}
Here $\mathcal{T}(t)$ is the tube of radius $\ep$ around the light particle's trajectory: 
\begin{equation*}
\mathcal{T}(t) := \{ y : \exists s \in [0,t] \text{ such that } \abs{ y - x(s) } \leq \ep \}.
\end{equation*}
Then we have again an asymptotically vanishing error:
\begin{equation*}
f^\ep_\alpha(t,x,v) = e^{- \rho^\ep_\alpha \abs{ \mathcal{T}(t) } } \sum_{M \geq 0} \frac{ (\rho^\ep_\alpha)^M }{ M! } \int_{B(x)^N} \chi_1 \chi_2 (\omega) f_0 ( \Phi^{t}_{\alpha, \omega, \ep} (x,v) ) d\omega + \phi_2(\ep).
\end{equation*}
Next we make the key observation that single obstacles dominate the trajectory's path. First, we note that there is a sequence of thresholds, $\alpha_n$ converging to $1/2$ from below, such that for $\alpha < \alpha_n$, there is a negligible number of clusters of $n$ obstacles. For instance, $\alpha_2 = 1/4$, that is, only for $\alpha \geq 1/4$ do we need to worry about the numbers (if not influence) of doublets (clusters of $2$ obstacles). 

Let $n$ be the number of internal doublets up to time $t$; $n$ is a random variable with an expected value of order $\ep^{1-4\alpha}$. Let $\theta_j$ be the deflection angle for the collision with the $j$-th doublet.

Then the key observation is that $\mathbb{E}^\ep [\sum_{j=1}^n \theta_j] \to 0$ as $\ep \to 0$, and can be seen as follows. We introduce an expansion of the deflection angle at the $j$-th doublet (which comes from \eqref{thetaexpansion} in the next section): 
\begin{equation*}
 \theta_j = \ep^\alpha A^j_1 + \ep^{2\alpha} A^j_2 + O(\ep^{3\alpha}).
\end{equation*} 

We note that $\mathbb{E} \theta_j = 0$, even when the test particle hits a cluster of $n_j$ scatterers overlapping to form one large obstacle. This can be seen easily for the cluster which is actually just one scatterer, because $\theta_j$ is an odd function of the impact parameter. For $n_j > 1$, we assume without loss of generality that after leaving the previous cluster, the light particle's position is $x_j = 0$, and its velocity, $v_j = e_1$ (the first coordinate vector).  Let $r_i$ be the centers of the scatterers in the order that they are encountered, with $i$ between $1$ and $n_j$.  For each such configuration $R_j := \{r_i\}_{i = 1}^{n_j}$, there is an opposite configuration $\tilde{R}_j$, obtained by reflecting across the first coordinate axis, with corresponding scattering angle $\tilde{\theta}_j = -\theta_j$.  So $\theta_j$ is an odd function in this sense, and vanishes when integrated against an even measure.

Then we compute:
\begin{equation*}
\begin{split}
 \mathbb{E}^\ep \biggl[\sum_{j=1}^n \theta_j \biggr] & = \mathbb{E} \biggl[ \sum_{j=1}^n \ep^\alpha A^j_1+ \ep^{2\alpha} A^j_2 + O(\ep^{3\alpha})\biggr]\\
& = \mathbb{E} \biggl[ \sum_{j=1}^n \ep^{2\alpha} A^j_2 \biggr] + O(\ep^{3\alpha}) \\
& = \mathbb{E}[n] \mathbb{E}\bigl[ \ep^{2\alpha} A^j_2\bigr] \\
& = O(\ep^{1-4\alpha}) O(\ep^{2\alpha}).
\end{split}
\end{equation*}
And this last quantity goes to zero as $\ep \to 0$ because $\alpha \in (0, 1/2)$.

\begin{figure}[htpb]
      \centering
      \includegraphics[viewport=0in 0in 7.5in 3.5in, keepaspectratio, width=5in,clip]{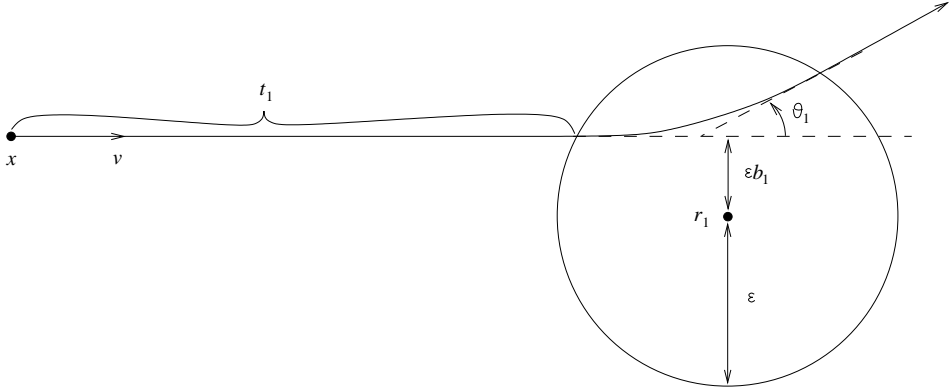}
      \caption{The changes of variables, from $r_j$ to $t_j$ and $b_j$, and from $b_j$ to $\theta_j$.}
      \label{softchangevars}
\end{figure}

Resuming along the lines of \cite{DR}, we introduce a change of variables $\mathcal{L}$, replacing obstacle locations by hitting times and impact parameters:
\begin{equation}
 \mathcal{L}: \{ r_j \}_{j=1}^n \mapsto \{ t_j, b_j \}_{j = 1}^n.
\end{equation} 
The difference is that here, we define the domain of $\mathcal{L}$ to be the set of all trajectories that do not start on a scatterer; that do not have extraneous scatterers; and that have stopping time $\tau^\ep > t$ (i.e., they have no almost-tangential self-intersections, not too many self-intersections, and velocity that does not change too much). In particular, we handle the change of variables at a doublet as illustrated in figure \ref{changevarsdoublet}. 

\begin{figure}[htpb]
      \centering
      \includegraphics[viewport=0in 0in 5.1in 2.9in, keepaspectratio, width=4in,clip]{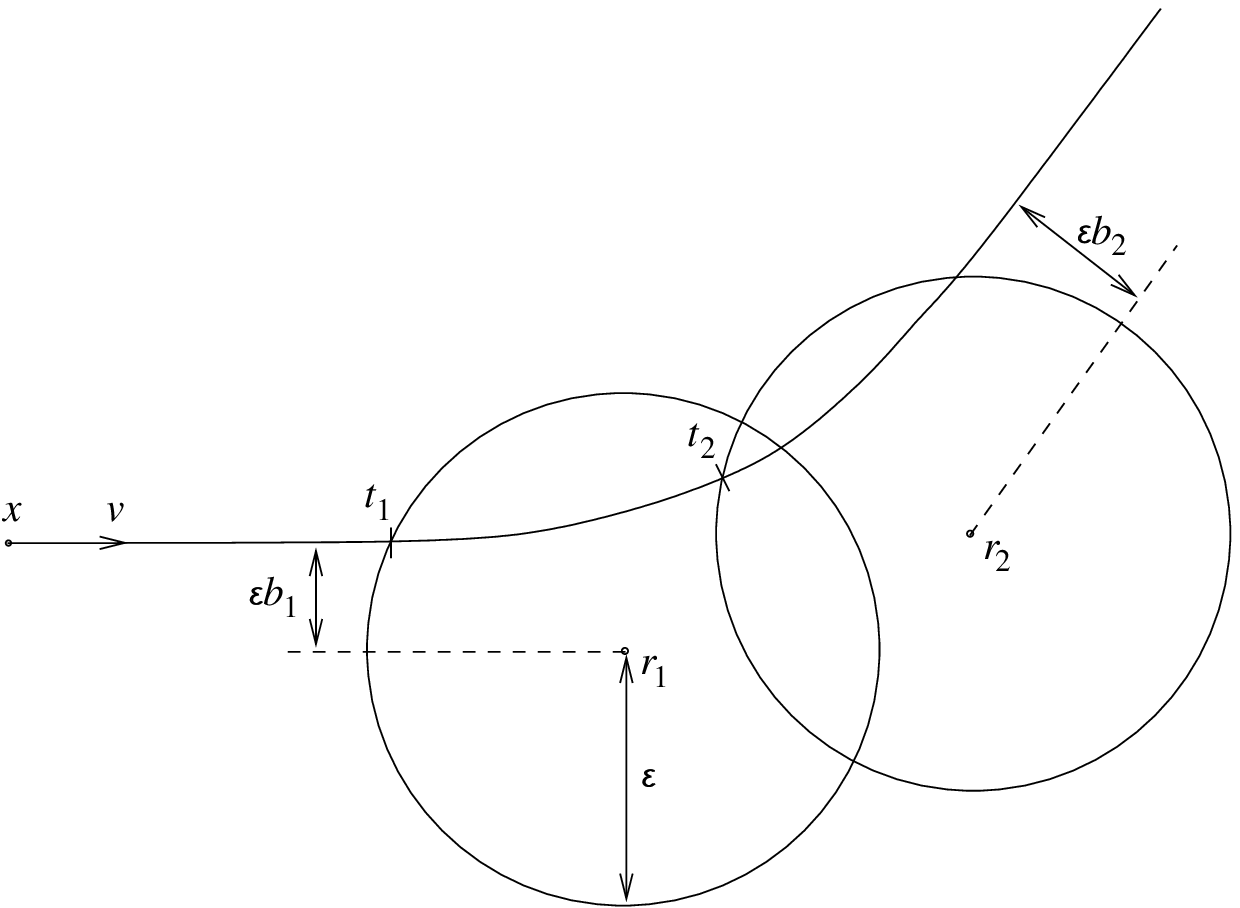}
      \caption{The change of variables for a doublet.}
      \label{changevarsdoublet}
\end{figure}

Then we can write $f^\ep_\alpha$ using the change of variables and \eqref{stoppingtimeproperties}, with an error $\phi_3(\ep)$ vanishing as $\ep \to 0$, as:
\begin{equation*}
e^{- \rho^\ep_\alpha \abs{ \mathcal{T}(t) } } \sum_{M \geq 0} (\rho^\ep_\alpha)^M  \int_{\triangle} \int_{\square_\ep}  \chi(\{ t_i, b_i \}_{i = 1}^M  \notin Range(\mathcal{L}) ) f_0 ( \Phi^{t}_{\alpha, \omega, \ep} (x,v) ) d\bar{b} d\bar{t}  + \phi_3(\ep).
\end{equation*}
Here $\bar{t}$ stands in for $(t_1, \dots, t_k)$, and $\triangle := \{ \bar{t} : t_1 \in (0,t), t_2 \in (t_1,t), \dots, t_k \in (t_{k-1},t) \}$. Similarly, $\bar{b} := (b_1, \dots, b_k)$, and $\square_\ep := (-\ep,\ep) \times \cdots \times (-\ep, \ep)$.

Then the procedure of \cite{DR} can be followed and a change of variables made (with a slight modification for doublets, described below), from impact parameters $\{ b_i \}$ to deflection angles $\{ \theta_i \}$ (see figure \ref{softchangevars}), which has Jacobian determinant:
\begin{equation*}
 \prod_{i=1}^M \ep^{1+2\alpha} \Gamma_\ep(\theta_i) := \prod_{i=1}^M \frac{ db_i }{ d\theta_i }.
\end{equation*}
Here, $\Gamma_\ep(\theta_i)$ is the rescaled scattering cross section.
We use the fact that the deflection angle $\theta$ through a doublet can be approximated: $\theta = \theta_{1} + \theta_{2} + \phi(\ep).$ Here $\theta_{1}$ and $\theta_{2}$ are the deflection angles corresponding to $b_1$ and $b_2$ in figure \ref{changevarsdoublet}, and $\phi(\ep)$ is a small error vanishing as $\ep \to 0$.

Using also $R_\theta (v)$ to denote the rotation of the vector $v$ by angle $\theta$, we make the following polygonal approximation to the trajectory:
\begin{equation*}
x(t) = x + \sum_{i = 0}^M R_{\theta_1 + \cdots + \theta_i} (v) (t_{i+1} - t_i) + O(M\ep).
\end{equation*} 

Additionally, we approximate $\abs{\mathcal{T}(t)}$ by $2 \ep t$. Then we can rewrite $f^\ep_\alpha$:
\begin{equation*}
\begin{split}
f^\ep_\alpha(t,x,v)  &= e^{- \rho \ep^{-1-2\alpha} 2 \ep t } \sum_{M \geq 0} \rho^M (\ep^{-1-2\alpha})^M  \int_{\triangle} \int_{\square_\pi}  \prod_{i=1}^M \ep^{1+2\alpha} \Gamma_\ep(\theta_i) \\
& \quad \quad \quad \times f_0 \biggl( x + \sum_{i = 0}^M R_{\theta_1 + \cdots + \theta_i} (v) (t_{i+1} - t_i), R_{\theta_1 + \cdots + \theta_M}(v) \biggr) d\bar{\theta} d\bar{t} + \phi_4(\ep)\\
&= e^{- t \rho \int_{-\pi}^{\pi} \Gamma_\ep(\theta) d\theta } \sum_{M \geq 0} \rho^M  \int_{\triangle} \int_{\square_\pi}  \prod_{i=1}^M \Gamma_\ep(\theta_i) \\
& \quad \quad \quad \times  f_0 \biggl( x + \sum_{i = 0}^M R_{\theta_1 + \cdots + \theta_i} (v) (t_{i+1} - t_i), R_{\theta_1 + \cdots + \theta_M}(v) \biggr) d\bar{\theta} d\bar{t} + \phi_4(\ep)
\end{split}
\end{equation*}
We then recognize this expansion as the series form of a solution to the family of Boltzmann equations:
\begin{equation*}
\begin{split}
	(\partial_t + v \cdot \nabla_x) h_\ep(t,x,v) & = \rho \int_{-\pi}^{\pi} \Gamma_\ep(\theta)[h_\ep (t,x, R_\theta(v)) - h_\ep(t,x,v)] d\theta \\
	h_\ep(0,x,v) & = f_0(x,v).
\end{split}
\end{equation*}
Finally, the $h_\ep$ converge in the appropriate sense to $h$, the solution of the Landau equation \eqref{LinearLandau}, because the scattering cross sections $\Gamma_\ep$ concentrate on grazing collisions.

We note that $\mathbb{E} \theta_j = 0$, even when the test particle hits a cluster of $n_j$ scatterers overlapping to form one large obstacle. This can be seen easily for the case of hitting a lone scatterer, because $\theta_j$ is an odd function of the impact parameter. For $n_j > 1$, we assume without loss of generality that after leaving the previous cluster of scatterers, the light particle's position is $x_j = 0$, and its velocity, $v_j = e_1$ (the first coordinate vector).  Let $r_i$ be the centers of the scatterers in the order that they are encountered, with $i$ between $1$ and $n_j$.  For each such configuration $R_j := \{r_i\}_{i = 1}^{n_j}$, there is an opposite configuration $\tilde{R}_j$, obtained by reflecting across the first coordinate axis, with corresponding scattering angle $\tilde{\theta}_j = -\theta_j$.  So $\theta_j$ is an odd function in this sense, and when integrated against an even measure, vanishes.

\section{The diffusion constant}

We take the following as the definition of the diffusion constant:
\begin{equation}\label{defnzeta}
\zeta := \lim_{\ep \to 0} \frac{\rho}{2} \int_{-\pi}^{\pi} \theta^2 \Gamma_\ep(\theta) d\theta = \lim_{\ep \to 0} \ep^{- 2\alpha} \frac{\rho}{2} \int_{-1}^{1} \theta(b)^2 db.
\end{equation}

\begin{proposition}
The diffusion constant, $\zeta$, defined above, is independent of $\alpha \in (0, 1/2)$ and can be expressed by the following formula:
\begin{equation}\label{zeta}
\zeta = \frac{\rho}{2} \int_{-1}^1 \left( \int_b^1 V' \left( \frac{\abs{b}}{u} \right) \frac{b}{u} \frac{du}{\sqrt{1-u^2}} \right)^2 db.
\end{equation}
\end{proposition}

\begin{proof}
We do several expansions to compute the diffusion constant. First, we compare the deflection through one obstacle, centered at the origin, to the path that would be taken if the obstacle weren't there. Let the entry time, place, and velocity be $0$, $x^- = x(0)$, and $v^-$; the deflected exit time, position, and velocity be $\tau$, $x^+$, $v^+$; and the non-deflected (straight-line) exit time, position, and velocity be $\hat{\tau}$, $\hat{x}^+$, and $\hat{v}^+$. (See figure \ref{comparestraight}.)

\begin{figure}[htpb]
      \centering
      \includegraphics[viewport=0in .8in 3in 2.3in, keepaspectratio, width=2.7in,clip]{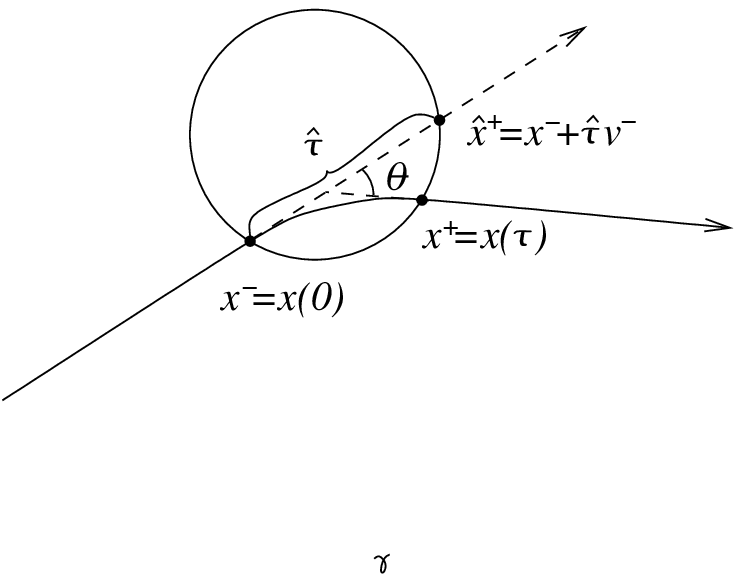}
      \caption{The straight-line trajectory exiting at time $\hat{\tau}$ and the deflected trajectory, at $\tau$, with the obstacle's center at the origin.}
      \label{comparestraight}
\end{figure}

We desire an estimate for $\abs{\tau - \hat{\tau}}$, so we examine the identity:
\begin{equation*}
\begin{split}
\ep^2 = \abs{\hat{x}^+}^2 & = \abs{x^-}^2 + \hat{\tau}^2 + 2 \tau x^- \cdot v^- \\
& = \ep^2  + \hat{\tau}^2 + 2 \tau x^- \cdot v^-.
\end{split}
\end{equation*}
This identity gives $\hat{\tau} = -2 x^- \cdot v^- = O (\ep)$, as we may assume that $\abs{v^-} = 1$. Similarly, $\tau = -2 x^- \cdot ( v^- + O(\ep^\alpha))$, so we have the estimate:
 \begin{equation}\label{tauhat}
 \abs{\hat{\tau} - \tau} = O(\ep^{1+\alpha}).
\end{equation}

Next, we estimate $v^+ - v^-$, using the definition $F^\ep_\alpha (y) := -\nabla V^\ep(y)$. We could make the following estimate (but we will actually do better):
\begin{equation*}
\begin{split}
v^+ - v^- & = \int_0^\tau F^\ep_\alpha (x(s)) ds \\
& = \int_0^{\hat{\tau}} F^\ep_\alpha (x(s)) ds +  \int_{\hat{\tau}}^\tau F^\ep_\alpha (x(s)) ds \\
& =  \int_0^{\hat{\tau}} F^\ep_\alpha (x(s)) ds + O(\ep^{1+\alpha}) \cdot O(\ep^{\alpha - 1}) \\
& = \int_0^{\hat{\tau}} F^\ep_\alpha (x(s)) ds + O(\ep^{2\alpha}).
\end{split}
\end{equation*}
Accordingly, we estimate that for $s$ between $\tau$ and $\hat{\tau}$, and some $a$ between $\tau$ and $s$, and using the compact support of $V$:
\begin{equation}\label{Fe}
\begin{split}
F^\ep_\alpha(x(s)) & = F^\ep_\alpha(x(\tau)) + \nabla F^\ep_\alpha(a)(x(s) - x(\tau)) \\
& = 0 + O(\ep^{ \alpha - 2}) \cdot O(\ep^{1+ \alpha}) \\
& = O(\ep^{-1 + 2\alpha}).
\end{split}
\end{equation}

We now combine \eqref{Fe} with \eqref{tauhat}, the estimate for $\abs{\hat{\tau} - \tau}$. Instead of the previous remainder term, $O(\ep^{2\alpha})$, we get:
\begin{equation*}
 \int_{\hat{\tau}}^\tau F^\ep_\alpha (x(s)) ds = O(\ep^{1+\alpha}) \cdot O(\ep^{-2 + \alpha}) \cdot O(\ep^{1+\alpha}) = O(\ep^{3\alpha}).
\end{equation*}
Thus:
\begin{equation}\label{vplusvminus}
v^+ - v^-  = \int_0^{\hat{\tau}} F^\ep_\alpha (x(s)) ds + O(\ep^{3\alpha}).
\end{equation}

Expand $F^\ep_\alpha(x(s))$, using $\hat{x}(s) = x^- + t v^-$, and rename the first two terms:
\begin{equation*}
\begin{split}
 F^\ep_\alpha (x(s)) & = F^\ep_\alpha (\hat{x}(s)) + \nabla F^\ep_\alpha (\hat{x}(s)) (x(s) - \hat{x}(s)) \\
 & \quad \quad + \frac{1}{2} D^2 F^\ep_\alpha (a) (x(s) - \hat{x}(s)) \cdot (x(s) - \hat{x}(s)) \\
& =: Y_1 + Y_2 + O(\ep^{-3+\alpha}) O(\ep^{1+ \alpha})^2 
\end{split}
\end{equation*}

To analyze $Y_2$ further, observe that:
\begin{equation*}
\begin{split}
 \nabla F^\ep_\alpha (\hat{x}(s)) (x(s) - \hat{x}(s)) & = \nabla F^\ep_\alpha (\hat{x}(s))
  \int_0^s [ v(t) - v^- ] dt  \\
& =  \nabla F^\ep_\alpha (\hat{x}(s))  \int_0^s \int_0^t F^\ep_\alpha(\hat{x}(u))du dt  \\
& \quad \quad + \nabla F^\ep_\alpha (\hat{x}(s))  \int_0^s \int_0^t [ F^\ep_\alpha(x(u)) - F^\ep_\alpha(\hat{x}(u)) ] du dt \\
& =: Y_{21} + Y_{22}.
\end{split}
\end{equation*}

Regarding the first term, $Y_{21}$, we integrate out $t$:
\begin{equation*}
\begin{split}
X_{21} : = \int_0^{\hat{\tau}} Y_{21} ds & = \int_0^{\hat{\tau}} \nabla F^\ep_\alpha (\hat{x}(s))  \int_0^s \int_0^t F^\ep_\alpha(\hat{x}(u)) du dt ds\\
& = \int_0^{\hat{\tau}} \int_0^s (s - u) \nabla F^\ep_\alpha (\hat{x}(s)) F^\ep_\alpha(\hat{x}(u)) du ds \\
& = O( \ep ^3) O( \ep^{ \alpha - 2})O( \ep^{ \alpha - 1}) = O( \ep^{ 2\alpha}).
 \end{split}
\end{equation*}
As for the second term in this decomposition, $Y_{22}$:
\begin{equation*}
\begin{split}
Y_{22} & =  \nabla F^\ep_\alpha (\hat{x}(s))  \int_0^s \int_0^t F^\ep_\alpha(x(u)) - F^\ep_\alpha(\hat{x}(u))du \\
& = O(\ep^{-2+\alpha}) O(\ep^2) O(\ep^{-2+\alpha} \ep^{1+\alpha}) \\
& = O(\ep^{-1 + 3\alpha}).
\end{split}
\end{equation*}
Then integrating, we get:
\begin{equation*}
X_{22} : = \int_0^{\hat{\tau}} Y_{22} ds = O(\ep^{3\alpha}).
\end{equation*}
Thus we have, combining the estimates for $Y_1$, $Y_{21}$, and $Y_{22}$:
\begin{equation*}
 v^+ - v^- = X_1 + X_{21} + O(\ep^{3\alpha}).
\end{equation*}

\begin{figure}[htpb]
      \centering
      \includegraphics[viewport=.5in .38in 3in 2.3in, keepaspectratio, width=1.9in,clip]{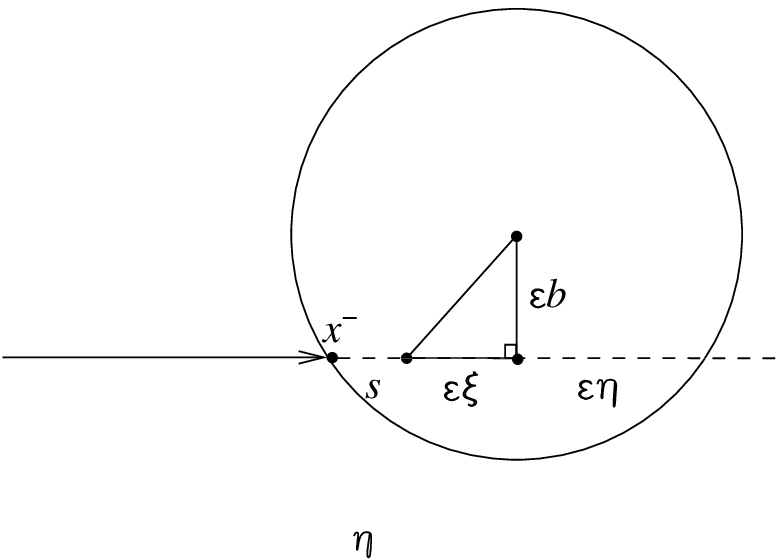}
      \caption{Change of variables from $s$ to $\xi$, which is the distance remaining to the center of the crossing (scaled by $\ep$).}
      \label{bxichangevars}
\end{figure}

Now we claim that $X_1 := \int_0^{\hat{\tau}} Y_1 = \int_0^{\hat{\tau}} F^\ep_\alpha(\hat{x}(s)) ds$ is perpendicular to $v^-$.
Without loss of generality we may assume that $v^- = (1,0).$ Making the following change of variables, $(s,b) \mapsto (\xi, b)$, as in figure \ref{bxichangevars}, we see that the claim is true:
\begin{equation*}
\begin{split}
 \int_0^{\hat{\tau}} Y_{1} & = - \int_0^{\hat{\tau}} V' \left( \frac{\abs{x^- + s v^-}}{\ep} \right) \frac{x^- + s v^-}{\abs{x^- + s v^-}} ds \\
& = - \ep^\alpha \int_{- \eta}^{ \eta} V'( \sqrt{b^2 + \xi^2} ) \cdot \frac{1}{\sqrt{b^2 + \xi^2}} (\xi, b) d\xi.
 \end{split}
\end{equation*}
Then we introduce a variable $u$ such that:
\begin{equation*}
 \frac{\abs{b}}{u} = \sqrt{b^2 + \xi^2},
  \quad \text{and thus} \quad \frac{\abs{b}}{u^2} du = \frac{- \abs{b} \xi d\xi}{ \sqrt{b^2 + \xi^2} }.
 \end{equation*}
 Using this change of variables, and $u_0 := \abs{b} / \sqrt{\eta^2 + b^2} = \abs{b}$,
\begin{equation}\label{X1bound}
\begin{split}
 \bigabs{ \int_0^{\hat{\tau}} Y_{1} } & = \bigabs{ \ep^\alpha \int_{-\eta}^{\eta} V'( \sqrt{b^2 + \xi^2} ) \cdot \frac{1}{\sqrt{b^2 + \xi^2}} b d\xi } \\
 & = \bigabs{ 2 \ep^\alpha \int_{0}^{\eta} V'( \sqrt{b^2 + \xi^2} ) \cdot \frac{1}{\sqrt{b^2 + \xi^2}} b d\xi } \\
& = \bigabs{ 2 \ep^\alpha \int_{1}^{u_0} V' \left( \frac{\abs{b}}{u} \right)   \frac{u} {\abs{b}}  b  \left( - \frac{\abs{b}^2}{u^3}\frac{u}{\abs{b} \sqrt{1 - u^2} } \right) du } \\
& = \bigabs{ 2 \ep^\alpha \int_{u_0}^{1} V' \left( \frac{\abs{b}}{u} \right) \frac{b}{u} \frac{du}{ \sqrt{1 - u^2} } }.
\end{split}
\end{equation}

So if we examine the deflection angle, $\theta$, as pictured in figure \ref{comparestraight}, we obtain:
\begin{equation*}
 1 - \frac{1}{2} \theta^2 + O(\ep^{3\alpha}) = \cos{\theta} = v^+ \cdot v^- = (v^+ - v^-) \cdot v^- + 1.
 \end{equation*}

Then orthogonality of $X_1$ with $v^-$ implies:
\begin{equation}\label{thetaexpansion}
\begin{split}
\theta^2 & = 2 ( X_1 + X_{21} ) \cdot v^- + O(\ep^{3\alpha}) \\
& = - 2 \int_0^{\hat{\tau}} v^- \cdot \nabla F^\ep_\alpha (\hat{x}(s))  \int_0^s (s - u) F^\ep_\alpha(\hat{x}(u)) du ds  \\
& = - 2 \int_0^{\hat{\tau}} \frac{d}{ds} F^\ep_\alpha (\hat{x}(s))  \int_0^s (s - u)  F^\ep_\alpha(\hat{x}(u))du ds  \\
& = - 2 \left[ F^\ep_\alpha (\hat{x} (s)) \int_0^s (s - u) F^\ep_\alpha(\hat{x}(u)) du \right]_0^{\hat{\tau}} \\
& \quad \quad + \int_0^{\hat{\tau}} F^\ep_\alpha (\hat{x}(s)) \frac{d}{ds} \int_0^s (s - u) F^\ep_\alpha(\hat{x}(u)) du  ds  \\
& = 0 + 2 \int_0^{\hat{\tau}} F^\ep_\alpha (\hat{x}(s)) \int_0^s F^\ep_\alpha(\hat{x}(u)) du ds \\
& =  2 \int_0^{\hat{\tau}} F^\ep_\alpha (\hat{x}(s)) F^\ep_\alpha(\hat{x}(u)) \chi_{ \{u < s\} } (u, s) du ds\\
& = \left \lvert \int_0^{\hat{\tau}} F^\ep_\alpha (\hat{x}(s)) ds \right \rvert^2 \\
& = \abs{ X_1 }^2.
\end{split}
\end{equation}
We note that in particular, \eqref{thetaexpansion} gives us an expansion of the deflection angle with leading term of order $\ep^{\alpha}$. In terms of the difference in velocities:
\begin{equation}
\begin{split}
v^+ - v^- & = X_1 + X_{21} + O(\ep^{3\alpha}) \\
& = X_1 - \frac{1}{2} \abs{X_1}^2 +  O(\ep^{3\alpha}) \\
& = X_1 - \frac{1}{2} \theta^2 +  O(\ep^{3\alpha}) \\
& =  O(\ep^{\alpha}) +  O(\ep^{2\alpha}) +  O(\ep^{3\alpha}).
\end{split}
\end{equation}
The proof is finished by putting this expression for $\theta^2$ into the definition of the diffusion constant, \eqref{defnzeta}, and using \eqref{X1bound}. 
 \end{proof}

\begin{lemma}
The diffusion constant for $\alpha \in (0, 1/2)$ is the same as for the case $\alpha = 1/2$ in \cite{DGL}, whose formula via the martingale characterization was:
\begin{equation}\label{DGLdiffconst}
\zeta := \pi \rho \abs{ v(0) }^{-1} \int \abs{k}^2 \abs{\hat{V}(\abs{k}) }^2 d\abs{k}.
\end{equation} 
\end{lemma}

\begin{proof}
This can been seen by comparing the previous computation with \eqref{DGLdiffusion2}.
We take the quantity $\int (Y_1 + Y_{21})$ and integrate it with respect to $x^-$, change variables ($\tau = u - s$ and $r = x^- + v^- s$), and apply Plancherel's theorem:
\begin{equation*}
\begin{split}
\int \! \int_0^\tau (Y_1 + Y_{21}) ds dx^-  &= - \! \int \! \int_0^{\hat{\tau}}  \int_0^s (s - u) F^\ep_\alpha(x^- + uv^-) \nabla F^\ep_\alpha (x^- + s v^-) du ds dx^-  \\
& =  \int \int \tau F^\ep_\alpha(r + \tau v^-) \nabla F^\ep_\alpha (r) d\tau d^2 r \\
& =  \frac{1}{2} \int \int \left[ ( \nabla_p \cdot F(r + p\tau) ) F(r) d\tau d^2 r \right]_{ p = v^- } \\
& =  \frac{1}{2}  \nabla_p \int \int e^{i k \cdot p} (k \otimes k) d\tau d^2 k \upharpoonright_{ p = v^- } 
\end{split}
\end{equation*}
\end{proof}

\section{Appendix}

Retaining the same essential argument, the result of \cite{DR} (convergence in expectation of the evolution of the initial distribution $f_0$) can be stretched to include $\alpha \in [1/8, 1/4)$. The estimate that must be improved in order to achieve this is (40) in \cite[Lemma 1]{DR}. This section illustrates how to iterate their geometric method to get a tighter bound on $J^{ii}_{1,\ep}$, the error term estimating the probability of non-consecutive overlappings and recollisions, a bound that decays for $\alpha < 1/4$.  

First fix $n$ disjoint closed subintervals of $(0, \pi)$, each of the form 
\begin{equation*}
I_m := [\phi_m, \phi_m + \pi/2n].
\end{equation*}

We assume that $\phi_1 < \phi_2 < \dotsb < \phi_n$.  For this iterative method to work, $n$ must satisfy $\pi/2n > C \ep^\alpha \geq \abs{ \theta_k }$.  Thus for each $m = 1, \dotsc , n$, there exist $\{ h_m \}_{m = 1}^n$ such that
\begin{equation*}
 \sum_{k=1}^{h_m - 1} \theta_k \in I_m.
\end{equation*}

Our trajectory will have $n$ gaps where the times $t_{h_m}$ can vary.  All other times $t_k$ and all angles $\theta_k$ are fixed. 

The case $n = 2$, $I_1 = [\pi/8, 3 \pi /8]$, and $I_2 = [5 \pi/8, 7 \pi /8]$ is illustrated in figure \ref{iterated}.

\begin{figure}[htpb]
      \centering
      \includegraphics[viewport=1.3in 0in 7.3in 5in, keepaspectratio, width=5in,clip]{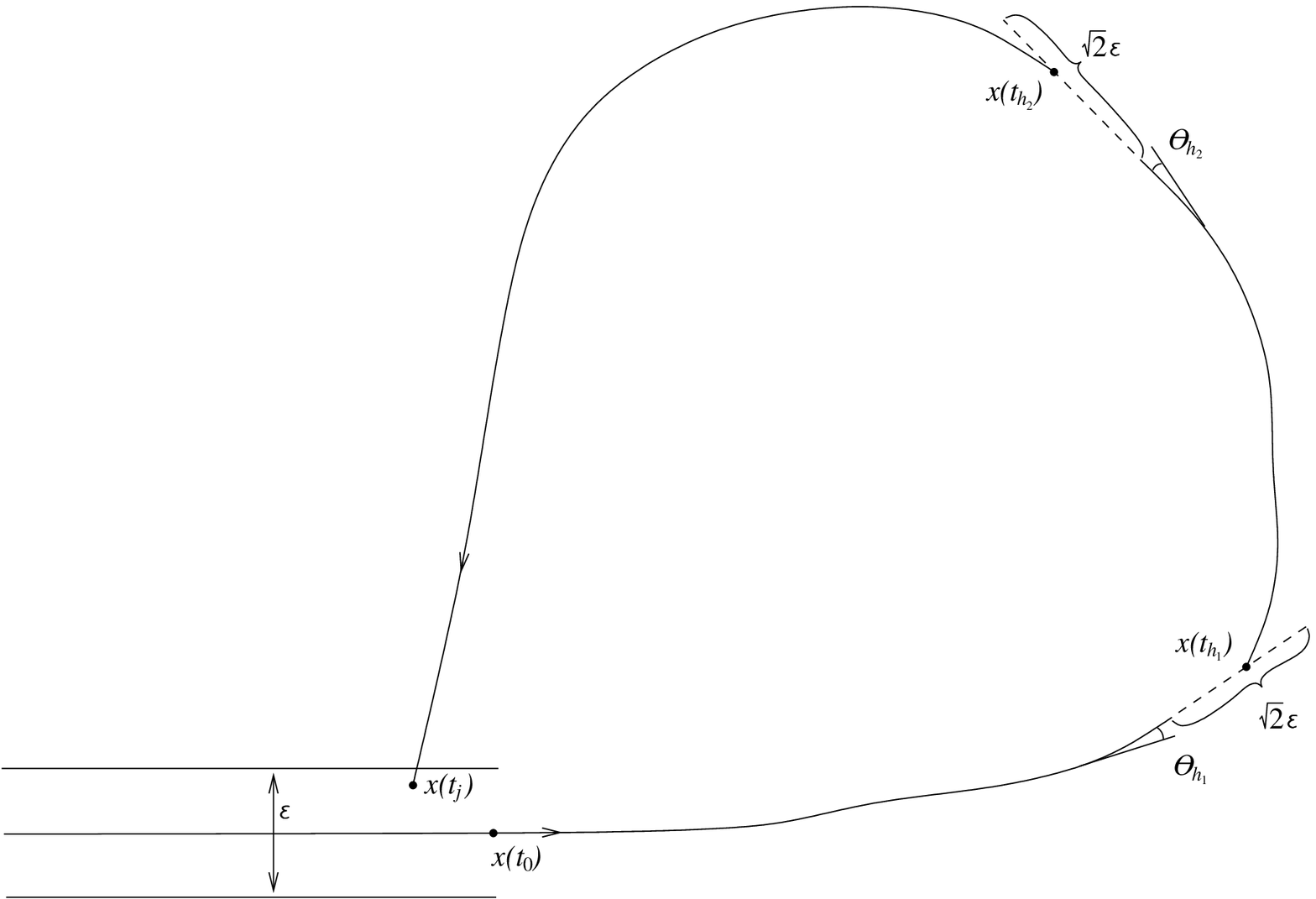}
      \caption{A recollision ($x(t_j)$ lying in the tube of radius $\ep$ around the trajectory) constrains how much $h_1$ and $h_2$ can vary in total; the variation for each is bounded by $\sqrt{2} \ep$.}
            \label{iterated}
\end{figure}

We claim that the recollision condition, $r_j \in \bigcup_{s \in (t_i, t_{i+1})} B(x(s), 2\ep)$, together with each $I_m$ being bounded away from the endpoints of $[0, \pi]$, results in each $t_{h_m}$ taking values in a set whose measure is $O(\ep)$.
The claim can be seen by working backwards from $h_n$:  the height of the trajectory at time $t_{h_n}$ can vary only in an interval of size $\ep$, due to the recollision condition restricting the trajectory to a tube of width $\ep$.  Then the total variation of all the $t_{h_m}$, $m = 1, \dotsc , n$ can be no larger than $C\ep$.  Here the constant $C$ is chosen so that 
\begin{equation*}
\left\lvert \frac{1}{\sin \phi} \right\rvert \leq C, \text{ for } \phi \in \{ \phi_1, \phi_n + \pi/2n \}.
\end{equation*}

Hence the variation of any single $t_{h_m}$ is bounded above by $C\ep$, there being no negative values of $t_{h_m}$ allowed and all angles pointing upwards at the variable times (i.e., being strictly between zero and $\pi$).  Then $J^{ii}_{1,\ep}$ can be estimated like before:

\begin{equation*}
\begin{split}
J^{ii}_{1,\ep} &\leq e^{-2t\ep \rho^\ep_\alpha}  \sum_{Q \geq 1} (\rho^\ep_\alpha)^Q 
\int_0^t dt_1 \dotsi \int_{t_{Q-1}}^t dt_Q \int_{-\ep}^\ep d\rho_1 \dotsi \int_{-\ep}^\ep d\rho_Q  \sum_{i = 0}^{Q-1} \sum_{j = i+2}^{Q}  \sum_{h_1 = i+1}^j \dotsb \\
& \quad \sum_{h_2 = i + 2}^j \dotsb \sum_{h_n = i + n}^j
\prod_{m = 1}^n \mathbf{1} \biggl( \Biggl\{ \sum_{k=1}^{h_m - 1} \theta_k \in I_m \Biggr\} \biggr) \mathbf{1} \biggl( \Biggl\{ \beta_j \in \! \! \bigcup_{s \in (t_i, t_{i+1})} \! \! \! B(\xi(s), 2\ep)\Biggr\} \biggr) \\
& \leq e^{-2t\ep \rho^\ep_\alpha}  \sum_{Q \geq 1} \frac{(2 \ep \rho^\ep_\alpha)^Q}{(Q - n)!} Q^{n+2} t^{Q-1} C \ep^n \\
& \leq C(T) \ep^{ 3n + 2 - (2n + 2) \delta }  = C(T) \ep^{ n - 4(n+1) \alpha}.
\end{split}
\end{equation*}

The exponent is positive if $\alpha < \frac{n}{4(n+1)}$, so taking $n$ to infinity (as $\ep$ goes to zero) gives the result for $\alpha < 1/4$.

\bibliographystyle{amsalpha}

\end{document}